\documentclass[a4paper,12pt,english,leqno]{article}
\usepackage{amsmath,amsfonts,amssymb,amsthm,latexsym}
\usepackage{mathrsfs}
\usepackage{stmaryrd}
\usepackage{amscd}
\usepackage{mathtools}
\numberwithin{equation}{section}
\usepackage{natbib}

\usepackage[latin1]{inputenc}
\usepackage[T1]{fontenc}
\usepackage{dsfont}
\usepackage{verbatim}
\usepackage{enumerate}
\usepackage{xcolor}
\usepackage[pdftex,colorlinks,bookmarks=true,unicode=true,pdftitle={tree},pdfauthor={Wang and Wan}]{hyperref}%

\usepackage{ifpdf}
\usepackage{url}

\ifpdf
\usepackage[pdftex]{graphicx}
\else
\usepackage[dvips]{graphicx}\fi

\usepackage[scale=0.8]{geometry}
\title{Comeback kids: an evolutionary approach of the long-run innovation process\footnote{July 24, 2016}}
\author{Shidong Wang\footnote{\textit{Email:} \nolinkurl{shidong.wang@maths.ox.ac.uk}, S. Wang is supported by a EPSRC Grant EP/I01361X/1 at the University of Oxford} \\ {\scriptsize Mathematical Institute, University of Oxford}\\ {\scriptsize Andrew Wiles Building, Radcliffe Observatory Quarter, Woodstock Road, OX2 6CG, Oxford, United Kingdom}  \\~~\\
Renaud Foucart\footnote{ \textit{Email:} \nolinkurl{renaud.foucart@hu-berlin.de}  } \\ {\scriptsize Department of Economics, Humboldt University}\\ {\scriptsize
    Unter den Linden 6, 10099 Berlin , Germany}  \\~~\\
Cheng Wan\footnote{Corresponding author.  \textit{Email:} \nolinkurl{cheng.wan.2005@polytechnique.org}  } \\ {\scriptsize Department of Economics; Nuffield College, University of Oxford}\\ {\scriptsize New Road, OX1 1NF, Oxford, United Kingdom}}
 \date{ }

 \newcommand{\R}{\mathbb{R}}
 \newcommand{\X}{\mathcal{X}}
 \newcommand{\N}{\mathbb{N}}
 \newcommand{\PP}{\mathbb{P}}

 \newcommand{\al}{\alpha}
 \newcommand{\de}{\delta}

 \numberwithin{equation}{section}
 \newtheorem{theorem}{Theorem}[section]
 \newtheorem{proposition}[theorem]{Proposition}
 \newtheorem{lemma}[theorem]{Lemma}

 \newtheorem{definition}[theorem]{Definition}

\newtheorem{assumption}{Assumption}

 \hypersetup{
    linktocpage=true,
    citecolor=black,    
    linkcolor=black 
}

\begin{document}
\maketitle

\begin{abstract}
We provide a theoretical framework to understand when firms may benefit from exploiting previously abandoned technologies and brands. We model for the long run process of innovation, allowing for sustainable diversity and comebacks of old brands and technologies. We present two extensions to the logistic and Lotka-Volterra equations, which describe the diffusion of an innovation. First, we extend the short-term competition to a long-term process characterized by a sequence of innovations and substitutions. Second, by allowing the substitutions to be incomplete, we extend the one-dimensional process to a tree-form multidimensional one featuring diversification throughout the long-term development.
\end{abstract}

\bigskip

\newpage

\textit{There is an Indian proverb that goes, ``Sit on the bank of a river and wait: Your enemy's corpse will soon float by'' \footnote{Umberto Eco, ``the Name of the Rose,'' p.550}}

\section{Introduction}

The process of innovation is traditionally seen as a sequence of one-shot games, in which a new (and better) product or technology replaces an outdated standard. As discussed at length below, examples abound however, of declining brands and technological standards who came back from the dead following the arrival of a disruptive innovation. A very timely one is the ``comeback of vinyl.'' While this technology had been almost abandoned in the late nineties, vinyls are now the fastest growing area of music sales, with very positive future prospects, to the point that engineers predict that the vinyl record has come back to stay as the only analogue medium \citep{bartmanski2013vinyl}. 

While the appeal to nostalgia might be a driver of successful and temporary brand rejuvenation, little is known about the underlying process of innovation that make long lasting comebacks possible. We aim at providing a theoretical framework to understand a feature of innovation that is largely overlooked: when a new technology arrives and makes a previously dominant one obsolete, it may leave room for an increase in the demand for an even older technology.

We use stochastic dynamics in order to provide the micro-foundations for a deterministic modeling of the long-term process of innovation, competition and market diversification, based on the idea of ``close competitors.'' To keep the example of the vinyl, we claim to be able to recover the market evolution by the fact that, while compact discs were a direct competitor to vinyl, digital music is much less so. Consumers indeed claim to retrieve satisfaction from the consumption of a vinyl that cannot be obtained using a dematerialized device. Hence, as the market for CD gradually shrinks to be replaced by digitalized music (a direct competitor of CD), a market space opens for a limited but sustainable comeback of the vinyl. Studies in psychology, sociology and anthropology have shown that the choice to buy a new album in vinyl goes beyond simply being matters of nostalgia or fetish. It is precisely because music has become mostly dematerialized that more and more consumers buy new albums in what they perceive to be the most tangible format, the vinyl (\citealp{magaudda2011materiality}, \citealp{bartmanski2013vinyl}, \citealp{negus2015digital}). 

When a brand or technology becomes dominant and starts gaining market shares at the expense of a previously dominant one, our results suggest that one could benefit from going a step backward, and wonder what the previously dominant one replaced. If this abandoned brand or technology vanished due to better substitute features that are also present in the new dominant one, there is no hope for a long lasting return. If the new dominant brand or technology does have such features, a successful comeback is possible.

 Extensive research in business and economics has been carried out on the diffusion of innovation, with objective to understand  the diffusion of new technologies over the course of industrial history  (\citealp{Griliches1957,Griliches1960}, \citealp{Mansfield1961}, \citealp{Geroski2000}, \citealp{Young2009}, \citealp{PMM2010}). More recent models in business and economics have studied simultaneous launches and coexistence of technologies (see for instance \citealp{libai2009role}, \citealp{yan2011competitive} and \citealp{guseo2014within}), including applications to the case of the music industry \citep{guidolin2015technological}. None of them however studies the conditions for a fading technology to come back.

In other contexts, the subject of innovation can be a product,  an idea, a practice, a phenomenon, a social norm or convention, or even a religion. In this paper, we call the subject in question an ``alternative.'' For example, market analysts wish to predict and influence how a new product gradually occupies the market or how old ones vanish from it (\citealp{MahajanMuller1979}, \citealp{MMB1990}, \citealp{Parker1994}, \citealp{ChanTellis2007}); communication specialists try to understand the mechanism of diffusion in order to intervene in the propagation of an idea or information; anthropologists and sociologists investigate on how a particular practice spreads from one tribe, culture or region to another (\citealp{KLH1963}, \citealp{Brown1981}, \citealp{Clark1984}, \citealp{Rogers2003}); and political scientists are interested in the adoption of new policies among different states.

In most of the cases, the innovation does not arise out of nothing. One or several older alternatives with similar functions may well have existed before the arrival of the innovation, though the innovation can be more advanced or sophisticated and ready to be accepted by a larger population. In this sense, we talk rather of ``substitution'' than ``diffusion.'' For example, synthetic fibers replaced natural fibers in cloth making, and the diesel locomotive replaced the steam locomotive in railroad transportation. 

The seminal work of \cite{Mansfield1961} initiates the rigorous mathematical modeling of this substitution process, in particular in economics research. \cite{FisherPry1971} provides an empirical logistic S-type equation to describe the evolution of the market occupation rate of an innovation. A lot of studies follow this approach by extending and improving the model (\citealp{Blackman1972,Blackman1974}, \citealp{BretMahajan1980,EMM1981},  \citealp{SharifKabir1976a,SharifKabir1976b,SharifRama1981}, \citealp{Skiadas1985}, etc.)  A typical logistic differential equation is
\begin{equation}\label{eq:32}
\dot{n}=b n (N - n),
\end{equation}
where $n$ is the market share of the alternative in question, $N$ is the upper bound of the market share, $b$ is the coefficient of imitation or other mechanism leading to the adoption of the alternative by the population. In particular, $b$ can be a constant in a most elementary model or a function of $n$ in more sophisticated models.

Another approach consists in modeling the competition between the old and the new alternatives in an explicit way, by constructing a Lotka-Volterra dynamical system (\citealp{Batten1982}, \citealp{KBJ1985}, \citealp{Bhargava1989}). For example, let $m$ (resp. $n$) represent the number of individuals adopting the old alternative $x$ (resp. the new alternative $y$). Then, a system of Lotka-Volterra equations governing their competitive relation is:
\begin{equation}\label{eq:1}
\begin{cases}
\dot{m}=\alpha_1 m (M - m -\beta n) -\gamma_1 m, \\
\dot{n}=\alpha_2 n (N - n -\beta m) -\gamma_2 n.
\end{cases}
\end{equation}
Here, $\alpha_1$ is the rate of adoption, $\gamma_1$ the rate of abandonment and $M$ the upper bound of the market capacity for $x$; the notations are similar for $y$. Compared with the logistic equation \eqref{eq:32}, the parameter $\beta$ introduces, in addition, the competition between $x$ and $y$.

Both logistic and Lotka-Volterra equations were originally introduced in the study of animal population or epidemics,  where they remain an important tool in the study of populations (see for instance \citealp{lafferty2015general}). They are applied by analogy in the model of innovation diffusion and substitution in economics \citep{pistorius1997multi}, and are supported by empirical studies on the macroscopic level (the population level). In this framework, an innovation eventually becomes old, then a newer innovation appears and another round of competition takes place.

In actual markets, the process is however more subtle. The substitution of a new alternative for the old one is not always complete, and the coexistence of several generations of technologies is not rare. Moreover,  when the environment of the competition changes and, in particular, with the periodical arrival of innovations, the characteristic of the current competition may be influenced and changed as well.

We provide real world examples of such processes of coexistence and comebacks of alternatives in Section \ref{sec:examples}. We present in Section \ref{sec:short} an individual-based foundation of the logistic substitution and Lotka-Volterra competition processes, and then extend it to a {\em long-term} process characterized by a {\em sequence} of innovations, competitions and substitutions. This recovers the well-known representation of innovation as a succession of one-shot competitions between two alternatives (\citealp{Champnt2006}).  We then extend in Section \ref{sec:long} the model to a tree-form one, in which several alternatives survive after each short-term competition and, in the long run, innovations drive the growth of a ``development tree.''

We use a stochastic birth-death model as the foundation of the deterministic logistic and Lotka-Volterra equations. As remarked by \cite{FisherPry1971} on their logistic diffusion model, deterministic models are ``not to be applied to substitutions prior to their achieving a magnitude of a few percent, at which time a definite growth pattern is established and the very early history has little effect upon the trend extrapolation.'' The deterministic, continuous logistic equation and Lotka-Volterra dynamical system on the aggregate level is a macroscopic approximation of the stochastic, discrete process on the individual level. The approximation is therefore valid only if the size of the population is large enough so that the fluctuation or uncertainty on the individual level can be averaged.

\section{Examples \label{sec:examples}}
\subsection{The Vinyl Comeback}
According to Nielsen Soundscan, more than 9.2 million vinyl records were sold in the U.S. in 2014, marking a 52\% increase over the year before, the highest numbers recorded by SoundScan since the music industry monitor started tracking them back in 1991. These numbers are likely to be under-evaluated, as a large share of vinyl albums are purchased at independent record stores who do not necessarily report their figures.\footnote{The facts in this subsection are borrowed from: Eric Felten, ``It's Alive! Vinyl Makes a Comeback,'' the Wall Street Journal, January 27, 2012, Megan Gibson, ``Here's Why Music Lovers Are Turning to Vinyl and Dropping Digital,'' The Time, January 13, 2015, Allan Kozinn, ``Weaned on CDs, They're Reaching for Vinyl,'' The New York Times, June 9, 2013,  and  Glenn Peoples and Russ Crupnick, ``The True Story of How Vinyl Spun Its Way Back From Near-Extinction,'' Billboard, December 17, 2014} Depending on the sources, this represents between 6 and 15 percent of the total album sales on a year. These days, every major label and many smaller ones are releasing vinyl, and most major new releases have a vinyl version.

Analysts seem to all agree on one thing: the success of the vinyl is due to the fact that ``Records are admirably physical, the antithesis of the everywhere-and-nowhere airiness of \textit{the cloud},'' and that high involvement in music is connected to a perception of tangible records as more valuable \citep{styven2010need}. And, indeed, when talking about the advantages of vinyl, the benchmark of today's mass consumption of music is not the CD anymore, but digital music. In \textit{Billboard }magazine, Gleen Peoples and Russ Crupnick relate a research study carried on behalf of the Music Business Association in 2006, when the industry was trying to save the market for physical records from the competition of illegal downloads. Consumers were surveyed over a number of innovating products, such as interactive and connected CDs, branded cards with a download code, or music placed on branded storage devices. Among these alternatives was a ``premium'' vinyl including a digital download card. Of all these products, vinyl tested the worse, and didn't particularly appeal to any important fan segments. Almost ten years later, all these relatively successfully tested products became, in the best case, very small niches, while premium vinyl records are now widely sold.

Our claim is that in 2006 the music industry was looking for an alternative able to improve over the CD standard. Today's vinyl success comes from the fact that it is not so much compared to CD anymore, but mostly to digital music. In the words used in the paper, CD was a close competitor of vinyl, and a close competitor of digital music. However, vinyl is not a close competitor of digital music, and therefore as CD is supplanted by digital music, there is room for a comeback of vinyl.

\subsection{Lego and the return of ``traditional'' toys}

A private company can benefit from understanding the impact of new technologies on a possible comeback of previous standards. In 2003, the once flourishing Danish company Lego was virtually out of cash, on its way to bankruptcy.\footnote{The facts in this subsection are borrowed from: Craig McLean, ``Lego, play it again,'' The Telegraph, December 17, 2009, ``Innovation Almost Bankrupted Lego - Until It Rebuilt with a Better Blueprint,'' Time, July 23, 2012,  Jay Greene, ``How LEGO revived its brand,'' Business Week, July 23, 2010, and Nick Watt and Hana Karar, ``The land where Lego comes to life,'' Abc News, November 16, 2009.} The decay of the company began in the mid-nineties, when children started to abandon the traditional building blocks in favor of more sophisticated toys. It was a time of video games, electronic and flashy innovations, or cheaper and more disposable toys. The strategy of Lego was therefore to mimic this trend, producing a large variety of innovating products, without a lot of commercial success.

In 2005 however, the company decided to return to its basics, simplifying its product line by reducing the number of components produced in its factories from 12,400 to 7,000, and going back to its core principle of building bricks. This strategy seems to have largely paid. Lego is again a profitable company and, from 2012 to 2015, Lego's sales have gone up an average of 24\% annually and profits have grown by 41\%.

Our claim is that part of this success can be explained by innovations in the entertainment industry, leaving more and more children playing on dematerialized games, such as computers or tablets. Consumers in search of easy and innovative toys started to prefer those to a number of relatively short-living physical toys. However, this substitution opened a new avenue for long-lasting, high quality, physical toys such as Lego's, which are not ``just on a screen'' and no close competitors to tablets or software.

\subsection{Nuclear and renewable electricity}

In the early seventies, following the peak in the prices of fossil fuels, most developed countries started to massively invest in nuclear power plants.\footnote{The facts in this subsection are borrowed from \cite{neuhoff2005large}, \cite{verbruggen2008renewable}, \cite{sovacool2009intermittency},  Eurostat (2013) and data available on the website of the world nuclear association (world-nuclear.org)} The technology was rather recent and, at the time, relatively cheaper and socially acceptable, so that in the early 2000s, countries like France, Sweden or Belgium had more than $50\%$ of their electricity coming from nuclear sources (with almost $80\%$ for France). One of the drawbacks of nuclear power as a ``baseload power'' however, is that it takes time to turn on and off, and therefore needs to be complemented by more flexible sources of energy.

As environmental concerns become more and more important and the social acceptability of nuclear power decreases, renewable technologies are pushed forward in order to become competitive on the electricity market, with the ambition of making it the baseload power in a large number of developed countries. An important characteristic of renewable electricity sources is their intermittency, which makes them less and less compatible with nuclear power plants as soon as renewable electricity represents a higher share of the production. Hence, with the technologies available today, a direct consequence of higher shares of renewable is that it increases the cost of nuclear power and makes the more flexible sources such as gas power plants much more attractive, even more so than as a complement of nuclear power.

\subsection{The Nintendo 64}

A last example can be found in the release of the Nintendo 64 in the late nineties.\footnote{The facts in this subsection are borrowed from: Lawrence Fisher, ``Nintendo Delays Introduction Of Ultra 64 Video-Game Player,'' The New York Times, May 6, 1995, Mark Langshaw, ``Sony PlayStation vs Nintendo 64: Gaming's Greatest Rivalries,'' Digital Spy, December 9, 2012, Michael Krantz, ``Super Mario's dazzling comeback,'' Time, June 24, 2001 and Cesar Bacani and Murakami Mutsuko, ``Nintendo's new 64-bit platform sets off a scramble for market share,'' Asia Week, April 18, 1997.} Nintendo was a direct competitor of Sega and both were the remaining historical producers of video game consoles, sharing the common characteristics of having very identifiable characters following from one platform to the other (for instance, Mario for Nintendo and Sonic for Sega). In 1994, both Sega and Nintendo had planned to release a new platform. The Sega Saturn was innovating by providing a new technology: video games were produced on CD instead of rom cartridges, allowing for cheaper and much more flexible production of games, at the cost of being slower to load. However, two unexpected elements arrived. First, a new competitor (Sony), launched a high quality console based on the CD technology (the Playstation), which was quickly adopted as the new standard instead of the Sega Saturn. Second, for technical reasons, Nintendo had to delay the release of its new console. When the Nintendo 64 was released two years behind schedule, the Sega Saturn had failed to become a standard, and Nintendo was competing mostly against the Playstation. At this time, the two products were sufficiently differentiated - all the old characters and a more stable technology versus a more performance-based console - so that Nintendo managed to achieve significant market shares and a relative commercial success.  Hence, it is again the same logic: Nintendo has been able to survive because its close competitor Sega had been substituted by a non-close competitor, Sony.

\section{Sequential substitution model \label{sec:short}}
\subsection{Individual-based competition and substitution model}\label{section:preliminary}
This subsection constructs a stochastic birth-death model describing the short-term competition between alternatives on the individual level, which leads to the substitution of a new alternative (i.e. innovation) for the old one. This is the elementary model and the base of the paper. We discuss two particular examples, which correspond to the logistic diffusion model and the Lotka-Volterra competition model. The next subsection introduces a long-term process made up of successive innovations sparsely distributed along the time line, each followed by an elementary short-term competition/substitution process.

The mathematical tools used in this section are borrowed from population dynamics studied by probability theorists. We first describe the original model and adapt it to our context. Then we introduce and interpret the analytical properties of the dynamics.

Alternatives achieving the same goal or having similar functions compete with each other when they are both present in the market. An alternative is characterized by its \emph{trait}, whose value belongs to a given trait space $\X$. For example, coal, oil, gas, hydroelectric power and  nuclear have distinct traits as energy resources. An individual, who can be a person, a firm, etc., is a (potential) adopter of an alternative. We consider, for each trait, the number of individuals having adopted an alternative of that trait. To fit into the language of birth-death process, we call the event whereby an individual who has not yet adopted any alternative adopts an alternative of trait $x$  ``the birth of an individual of trait $x$.'' For example, it can be the entry of a firm into the market. Conversely, the ``death of an individual of trait $x$'' represents the event that an individual actually holding an alternative of trait $x$ abandons the trait while not adopting another one. For example, it can be the exit of a firm from the market, either for a reason linked only to its own performance, or because of its loss of market share due to competition against better-equipped firms. Finally, ``mutation'' means the birth of an individual with a new trait that has never existed before, which is interpreted as the arrival of an innovation.

Recall that in continuous-time birth-death processes, an individual gives birth or dies at independent, random exponential times. Besides, the reproduction rate, death rate and mutation rate vary among different traits.

Now let us consider a multi-trait birth-death process with mutations. At any time $t$, the population is composed of a finite number $I_t$ of individuals, respectively characterized by their trait $x_1(t), \ldots, x_{I_t}(t)$. The traits belong to a given closed subset of $\R^d$, called the trait space $\X$. Let $\delta_{x}$ stand for the Dirac distribution concentrated on $x$. The population state at time $t$ is specified by the counting measure on $\X$
\begin{equation*}
\nu_t=\sum\limits_{i=1}^{I_t}\delta_{X_i(t)}.
\end{equation*}

To write the dynamics of the process $(\nu_t)_{t>0}$, let us  introduce the following parameters:
\begin{itemize}
  \item $b(x)$ is the \emph{clonal birth} rate from an individual with trait $x$, who reproduces an offspring with the same trait as its parent.
  \item $d(x)$ is the \emph{natural death} rate of an individual with trait $x$.
  \item $\al(x,y)$ is the \emph{competition kernel} felt by some individual with trait $x$ from another individual with trait $y$, which results in the death of an individual with trait $x$.
  \item $\mu(x)$ is the \emph{mutant birth} rate of an parental individual with trait $x$. Mutation is indeed another type of birth event, which results in a newborn with a different trait from its parent.
  \item $p(x,dh)$ is the \emph{law of mutation variation} $h=y-x$ between a mutant $y$ and its parental trait $x$. Since the mutant trait $y=x+h$ is in $\X$, the law has support in $\X-x:=\{y-x: y\in\X\}\subset \R^d$.
\end{itemize}

The natural death is due to ``aging.'' For example, it can be the natural death of a person, or the exit of a firm from the market for reasons independent from the competition with other alternatives. The competition captures the influence of the presence of an alternative on the market on the survival rate of a given trait. For instance, if vinyl is directly threatened by CD, the massive presence of CDs in the market will increase the death rate of vinyl retailers. Competition is also present within an alternative: all other things being equal entry on the market for vinyl increases competition in this market. The total death rate of an individual of trait $x$ is hence $D(x) =d(x)+\int \al(x,y)\langle \nu_0, 1_{\{y\}}\rangle dy$, where $\langle \nu_t, 1_{\{y\}}\rangle dy$ is the size of individuals with trait $y$.

Let us give an algorithmic, pathwise construction of the dynamics of process $(\nu_t)_{t\geq 0}$. The construction can be summarized as follows.
\begin{description}
\item[(1)] At any given time, say initial time 0, choose an individual from the population $\nu_0$ at random. Assume that its trait is $x$.

\item[(2)] Simulate an exponential random variable $T_1$ of parameter $b(x)+D(x)+\mu(x)$.

\item[(3)] Simulate a uniform random variable $\theta_1$ on interval $[0, 1]$. According to its value, decide what happens to this individual of trait $x$ at time $T_1$: its death, a clonal birth from it or a birth with mutation from it.
\begin{itemize}
\item If $\theta_1\in[0, \tfrac{D(x)}{b(x)+D(x)+\mu(x)})$, the individual is killed and the sub-population size of trait $x$ decreases by one.
\item If $\theta_1\in[\tfrac{D(x)}{b(x)+D(x)+\mu(x)}, \tfrac{D(x)+b(x)}{b(x)+D(x)+\mu(x)})$, a clone birth of trait $x$ is reproduced and the sub-population size of trait $x$ increases by one.
 \item If $\theta_1\in[\tfrac{D(x)+b(x)}{b(x)+D(x)+\mu(x)}, 1]$, a birth with mutation occurs from parental trait $x$, and an innovative sub-population is created. The innovative trait $x+h$ is determined by the mutation transition law $p(x,dh)$.
\end{itemize}
\item[(4)] Return to step (1) and continue the iteration.
\end{description}

This construction serves as the definition of the individual-based model as well as a basic simulation algorithm for this model. The following assumption guarantees that the process is well defined.

\begin{assumption}\label{assp1}
There exist constants $\bar{b}$, $\bar{d}$, $\underline{\alpha}$, $\bar{\alpha}$, such that $~0<b(x)\leq\bar{b}$,  $0<d(x)\leq\bar{d}$, $0<\underline{\alpha}\leq\alpha(x,y)\leq\bar{\alpha}$, and $b(x)-d(x)>0$ for all $x$ in $\X$.
\end{assumption}
Indeed, owing to the upper-bounds of the transition rates given in Assumption \ref{assp1}, the dynamics are stochastically dominated by a Poisson process with birth rate $\bar{b}$ and death rate $\bar{d}+\bar{\al}\nu_t$. This fact justifies the existence of the process $(\nu_t)_{t\geq 0}$.

For each $K\in \mathbb{N}^*$, consider a population with initial size $K$. The proportion of individuals with each trait in the population is given by $X_t^K=\frac{\nu_t}{K}$. Our objective is to see whether the dynamics of $X^K$ can be approximated by deterministic equations when $K$ is large enough.

To give some flavor of our approach, let us cite two specific examples of process $(\nu_t)_{t\geq 0}$, respectively with only one trait and two traits in the trait space, and without further mutations, i.e. mutation rate $\mu(x)\equiv 0$ for any $x\in\X$. Hence the process depicts the short-term diffusion or competition behavior of the trait(s) on the individual level. It leads to an approximation of the usual logistic and Lotka-Volterra equations (cf. equations \eqref{eq:32} and \eqref{eq:1}) on the population level.

\begin{proposition}\label{prop:logistic}
Assume $\X=\{x\}$, $\mu(x)\equiv 0$, and $K\times\alpha(x, x)\to\alpha_0(x, x)$ as $K\rightarrow +\infty$. Suppose that the initial population size $\langle X^K_0, 1_{\{x\}}\rangle$ converges to $n_0$ as $K\rightarrow +\infty$, then $\langle X^K_t, 1_{\{x\}}\rangle$ is converges in probability to the solution of the following (logistic) equation:
\begin{equation}\label{eq:log}
\dot{n}_t(x)= \left( b(x)-d(x)-\alpha_0(x,x)n_t(x) \right) n_t(x),
\end{equation}
which has a unique stable point $\bar{n} (x)=\frac{b(x)-d(x)} {\alpha_0(x, x)} $.
\end{proposition}

The result has been proved in \cite[Theorem 5.3]{FM04} and is its one-dimensional realization. Similarly, there is a two-dimensional limiting process.

\begin{proposition}
Assume $\X=\{x, y\}$, $\mu(z)\equiv 0$ for $z \in \X$, and $K\times\alpha(w, z)\to\alpha_0(w, z)$ as $K\rightarrow +\infty$ for $w,z\in \X$. Suppose that the initial sub-population sizes $\langle X^K_0, 1_{\{x\}}\rangle$ and $\langle X^K_0, 1_{\{y\}}\rangle$ converges to $n_0(x)$ and $n_0(y)$ as $K\rightarrow +\infty$,  then $(\langle X^K_t, 1_{\{x\}}\rangle, \langle X^K_t, 1_{\{y\}}\rangle )$ converges in probability to the solution of the following (Lotka-Volterra) equations:
\begin{equation}\label{eq:lv}
\begin{cases}
\dot{n}_t(x)=\left(b(x)-d(x)-\alpha_0(x,x)n_t(x)-\alpha_0(x,y)n_t(y)\right)n_t(x), \\
\dot{n}_t(y)=\left(b(y)-d(y)-\alpha_0(y,x)n_t(x)-\alpha_0(y,y)n_t(y)\right)n_t(y).
\end{cases}
\end{equation}
\end{proposition}

\begin{proposition}\label{prop:stable analysis of LV}
Consider the Lotka-Volterra system $\left(n(x), n(y)\right)$ satisfying equations \eqref{eq:lv}.
Suppose that $n_0(x), n_0(y)>0$ and $\bar f(y,x):=b(y)-d(y)-\al(y,x)\bar{n}(x)>0$ and its symmetric form $\bar f(x,y)<0$. Then, we conclude that $(0, \bar n(y))$ is the only stable point.
\end{proposition}

The proof is given in Appendix.

\subsection{Development with sequential substitutions}\label{subsec:TSS_mutation}
This subsection takes mutation into consideration. This allows us to recover the well-known process consisting in a sequence of innovations, each followed by a short-term competition/substitution process introduced in the previous subsection \citep{Champnt2006}. In this model, innovations occur only rarely and do not take place on the same timescale as competition and substitution. In order to identify the proper range of mutation timescale in this sequential innovation-substitution model, let us replace the mutation rate $\mu(x)$ by $\sigma\mu(x)$. Besides, define  the \emph{fitness function} of trait $x$ with respect to $y$ by
\begin{equation*}
\bar f(x,y)=b(x)-d(x)-\alpha(x,y)\bar{n}(y).
\end{equation*}
The following assumption gives a non-coexistence condition for any pair of distinct traits.
\begin{assumption}\label{assp2}
For all $x,y$ in $\X$, $\bar f(x,y)\cdot \bar f(y,x)<0$.
\end{assumption}

Champagnat \cite[Theorem 1]{Champnt2006} proved the following result.
\begin{proposition}[\citealp{Champnt2006}]\label{TSS_Champagnat}
Consider a sequence of processes $(X^K_t)_{t\geq 0}$, $K\in \mathbb{N}^*$, with rescaled mutation law $\sigma\mu(x)$. Suppose that the initial points $X_0^K=\frac{N_0^K}{K} \delta_x$ satisfy that $\frac{N_0^K}{K}\stackrel{\text{law}}{\to} n_0$ as $K\rightarrow +\infty$, where $n_0$ is a strictly positive constant. Also suppose that, for all $C>0$,
\begin{equation}\label{Champagnat_condition}
\exp\{-CK\}\ll K\sigma \ll \frac{1}{\ln K}.
\end{equation}
Then, the trajectory $\{X^K_{t/K\sigma}\}_{ t\geq 0}$ converges in probability to $\{Y_t\}_{t\geq 0}$ such that
\begin{equation*}
Y_t=
  \begin{cases}
   n_0\delta_x, & t=0, \\
   \bar{n}(\eta_t)\delta_{\eta_t}, & t>0,
 \end{cases}
\end{equation*}
where $\{\eta_t\}_{t\geq 0}$ with $\eta_0=x$ is a Markov jump process in trait space $\X$ with jump rate from any $z\in \X$ to $y\in \X$
\begin{equation*}
          \bar{n}(z)\frac{[\bar f(y,z)]_+}{b(y)}m(z,dy).
\end{equation*}
Here $[\bar f(y,z)]_+$ denotes the positive part of $\bar f(y,z)\in \R$.
\end{proposition}

 \begin{figure}[hbtp]
 \centering
\includegraphics[width=220pt]{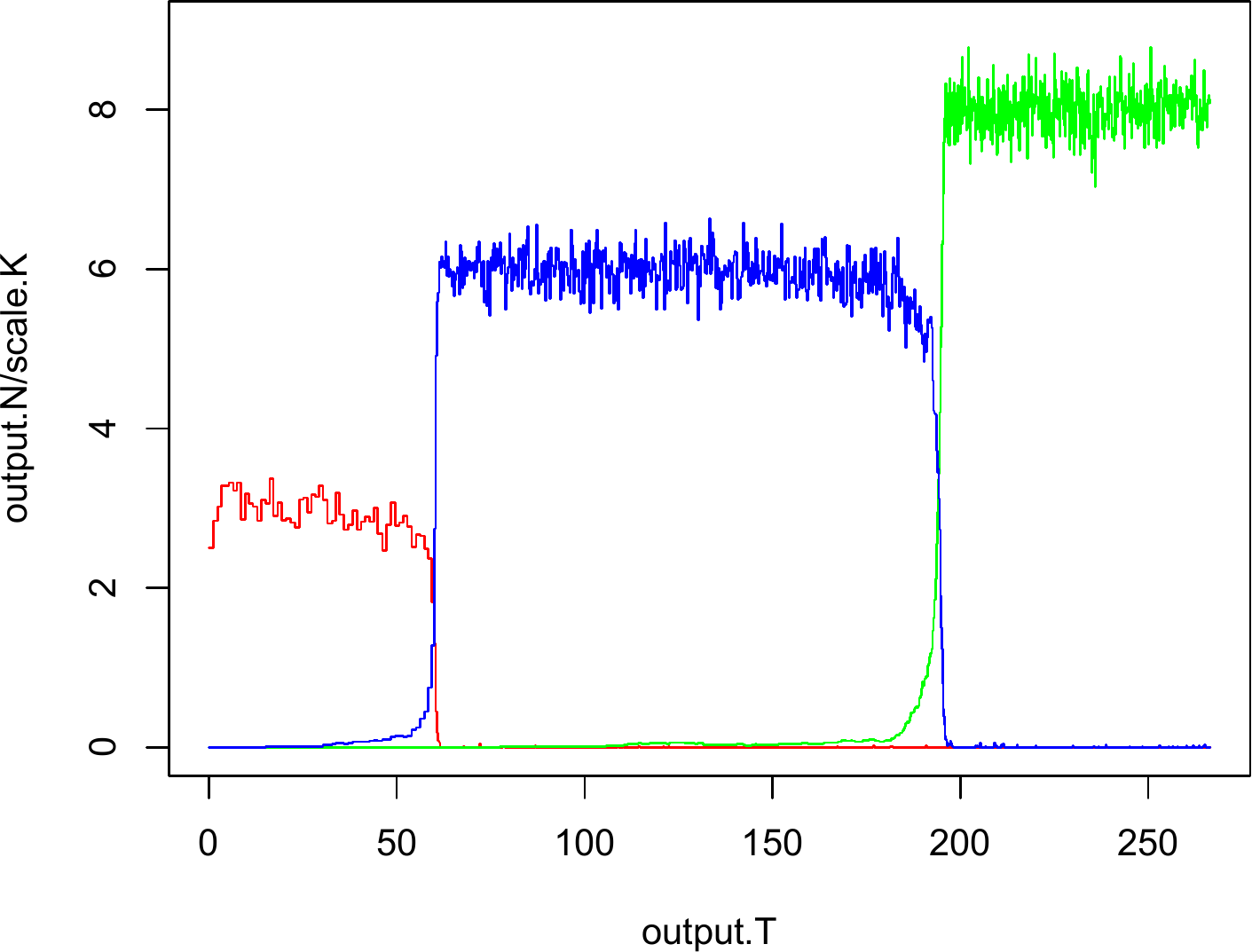}
\includegraphics[width=220pt]{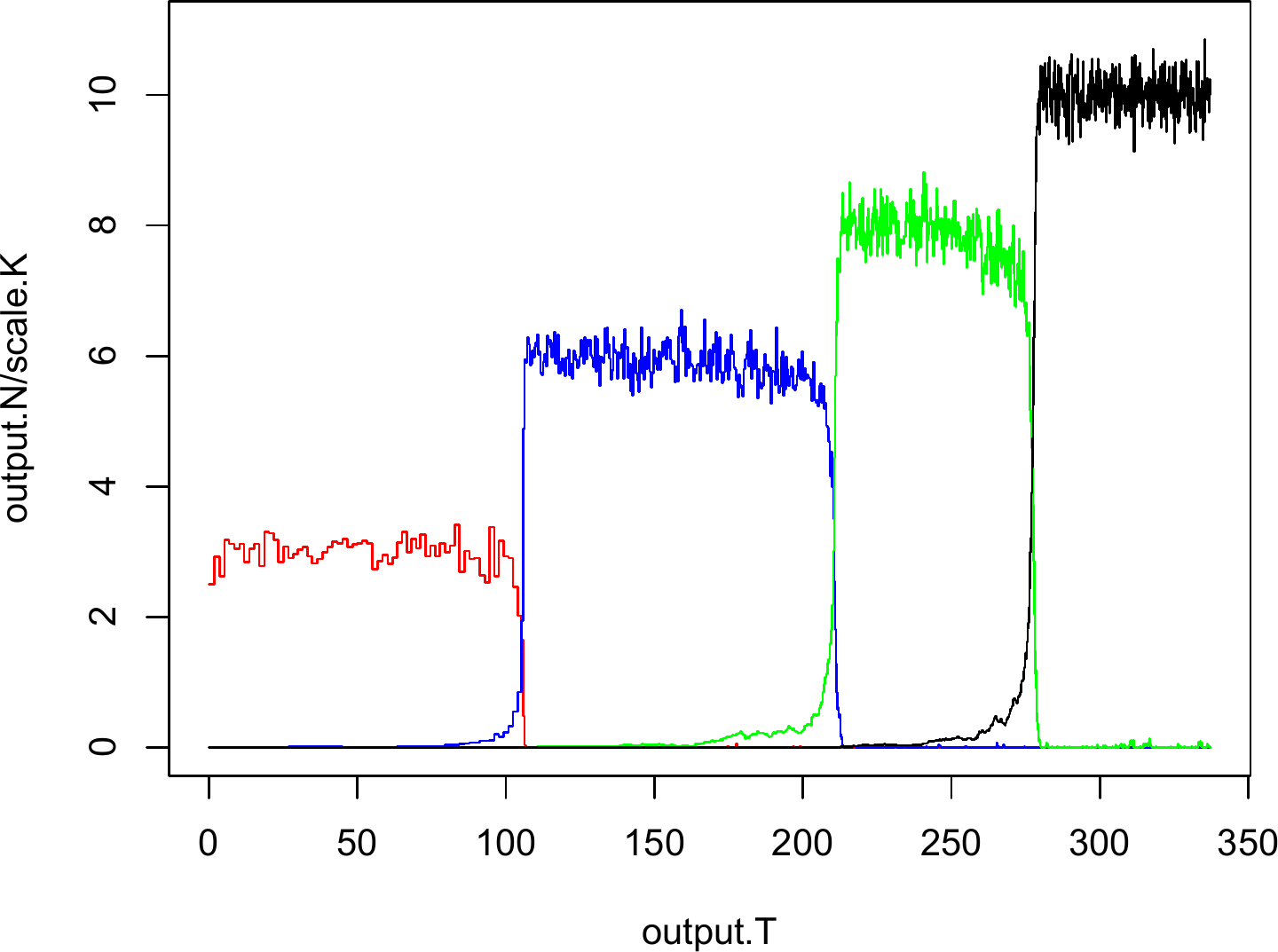}
\caption{\footnotesize Simulations of the monomorphic substitution sequence model arising in Proposition \ref{TSS_Champagnat}.}
\label{fig:TSS}
 \end{figure}

His arguments of proof gave rigorous mathematical justification of the following ideas. The timescale for the occurrence of mutations is of order $\frac{1}{K\sigma}$, whereas the fixation timescale (for the competition between a pair of traits to attain an equilibrium) is of order $\ln K$. According to \eqref{Champagnat_condition}, the interval between two occurrences of mutation is much longer than the fixation timescale.\footnote{However, $e^{-CK} \ll K\sigma$ guarantees that the mutation is not so rare that the fixed monomorphic population does not drift out of its unique stable equilibrium (cf. Freidlin and Wentzell \cite{FW84}).}  Proposition \ref{TSS_Champagnat} implies that, in this case, the population trait appears monomorphic (which means that all the individuals bear the same trait) at any moment on the mutation timescale. In other words, when a new trait appears by mutation, its competition with the existing one takes place very quickly and one of them is beaten and vanishes. Then all the individuals bear the same trait, namely the one that survives from the competition. This lasts for a long time until another mutation takes place.

Figure \ref{fig:TSS} shows some simulations of a sequential substitution model, with time on the x-axis and population on the y-axis . The trait space contains respectively three (resp. four) types for the left (resp. right) panel. The population densities of trait $x_0, x_1, x_2, x_3$ are marked respectively by red, blue, green and black colored curves. Take $b(x_0)=3$, $b(x_1)=6$, $b(x_2)=8$, $b(x_3)=10$ and death rates $d(x_i)\equiv 0$ for $i=0,1,2,3$. Other parameters are: competition kernel $\al\equiv 1$, migration kernel $m\equiv 0.5$, and mutation scale $\sigma=K^{-2}$, with initial population size $K=100$. One can see that the competition process of two successive alternatives, which leads to the fixation of the fitter and the vanishing of the less fit, takes very little time compared with the time interval between two innovations. Therefore, the typical double $S$-curves of the Lotka-Volterra equation \eqref{eq:lv} are almost invisible. One has to ``stretch'' the time line at each ``competition moment'' to recover them.

\section{Long-term development tree \label{sec:long}}
\subsection{Migration, local competition, diversification}
In the previous section, an individual always bears the same trait from its birth to its death. For example, when a firm enters the market and adopts a certain technology, it never changes that choice until it quits the market. However, before a firm with poor performance quits the market because it is losing market share to its competitors, it may well try to improve the performance by changing the technology. If it adopts its better performing competitor's technology, it is ``imitating''. If it tries some other technology randomly, it is doing a ``trial''. In the terminology of birth-death process models, this is called \emph{migration}.

In this section, by allowing migration of individuals, we extend the previous sequential substitution model into a tree-form development model. On the one hand, migration from an old trait to a new trait accelerates the diffusion of the new one in the population; on the other hand, migration from a new trait to an old trait prevents the old from completely vanishing. At any period, somebody will trial to revive an alternative, and depending on how competitive it is in the current market structure, it can make a comeback. Therefore, instead of having a ``monomorphic'' market all the time as in the sequential substitution model, several traits can coexist in the population under the refined model. In the long term, we obtain a ``development tree'' of the alternatives, with a diversity of sustainable alternatives.


Studying more than two alternatives at the same time brings another issue. When several alternatives coexist, does each of them compete with all of the others? Is there migration between each pair of alternatives? In this paper, we assume that the competition and the migration are rather local than global. For example, as means of transportation, bicycles compete with cars, cars compete with trains, trains compete with airplanes, but bicycles do not compete with airplanes. It is precisely because of this assumption that we can obtain comebacks in equilibrium: if the direct competitor of an alternative disappears, it is possible for this alternative to prosper again.

For these reasons, this paper considers only the competition and migration between traits with fitnesses that are ``close'' to each other. While in each specific context, an alternative may compete with several alternatives that are more or less close to it, we consider a particular case where each trait only compete with the two ``nearest neighbors'' in terms of their fitness. This makes the explicit form of the ``development tree'' tractable, and illustrates the main idea of the model.

In Section \ref{subsec:migration1}, we consider the case without mutation, in order to see how migration changes the process of short-term competition and substitution. Section \ref{subsec:migration2} extends the model to a long-term tree process by taking mutation, i.e. innovation, into consideration. Unlike in the previous section, where we apply established mathematical results to our context, in this section we develop new mathematical tools in probability theory in order to construct the tree model.

\subsection{Short-term competition model with migration}\label{subsec:migration1}
Let us introduce migration into the basic mutation-free model specified in Section \ref{section:preliminary}. Denote by $m(x, dy)$ the transition law for an individual migrating from trait $x$ to trait $y$. In this subsection, we assume that there is no mutation, $\mu(x)\equiv 0$. Suppose that competition and migration take place only between two alternatives that are neighbors on the \emph{fitness landscape}. Explicitly, for a given trait space $\X=\{x_0,x_1,x_2,\cdots,x_{L}\}$ containing $L+1$ distinct traits ($L\in \mathbb{N}$), admit the following assumption:

\begin{assumption}\label{assp3} The $L+1$ traits are ordered on an increasing fitness landscape, i.e. $x_0\prec x_1\prec\ldots\prec x_L$, where the order $\prec$ is defined as follows:\footnote{Note that the order is not a total order because it doesn't have the transitivity and comparability property. For instance, $\prec$ doesn't hold  between $x_1$ and $x_3$ even if $x_1\prec x_2$ and $x_2\prec x_3$.}
\begin{equation}\label{eq:order}
x \prec y \; \text{ if } \, \bar f(x , y)<0\,\text{ and }\,\bar f(y, x)>0.
\end{equation}

In addition,
\begin{equation*}
    m(x_i,x_j)=\al(x_i,x_j)=0,\quad \forall \mid i-j\mid>1.
\end{equation*}
\end{assumption}

Hence, trait $x_3$ competes with trait $x_2$, but not with trait $x_1$, and each trait is ``fitter'' than the previous one. We would like to study the birth-death process with such traits and local competition and migration behavior. Before stating Proposition \ref{TST}, let us introduce some notation and a technical assumption.

Recall that $\bar{n}(x)$ is the rest point of the logistic equation \eqref{eq:log} for trait $x$ (cf.  Proposition \ref{prop:logistic}). Following \cite{BovierWang2013}, define a specific configuration $\Gamma^{(L)}$ by
\begin{equation}\label{eq:gamma}
\Gamma^{(L)}:=
\begin{dcases}
\,\sum\limits_{i=0}^l\bar{n}(x_{2i})\de_{x_{2i}},\; &\text{ if }\, L=2l.\\
\,\sum\limits_{i=1}^{l+1}\bar{n}(x_{2i-1})\de_{x_{2i-1}}, &\text{ if }\, L=2l+1.
\end{dcases}
\end{equation}
In other words, when the traits are ordered by their fitness, at $\Gamma^{(L)}$, only every other trait is present, and its size is its equilibrium size with only intra-trait competition but not inter-trait competition.

\begin{proposition}\label{TST}
Consider the process $(X_t^K)_{t\geq 0}$ with rescaled migration law $\epsilon m(x, dy)$ on the trait space $\X=\{x_0,x_1,\ldots,x_L\}$. Suppose that for all $K$, $X_0^K=\frac{N_0^K}{K} \delta_{x_0}$, and $\frac{N_0^K}{K}\to n_0$ in law as $K\rightarrow +\infty$, where $n_0>0$. If
\begin{equation}\label{rare_migration}
1\ll K\epsilon\ll K,
\end{equation}
then, there exists a constant $\bar t_L>0$ such that,
\begin{equation*}
\lim\limits_{K\to\infty} X^K_{t\ln\frac{1}{\epsilon}}\stackrel{\text{(d)}}{=}\Gamma^{(L)}, \quad \forall\, t>\bar t_L,
\end{equation*}
under the total variation norm.
\end{proposition}

The proof is in the Appendix. If, in a large population, the frequency of migration events is much higher than that of birth and death events, but not high enough for the whole population to be migrating (cf. \ref{rare_migration}), then the outcome of the nearest-neighbor competition and migration is the elimination of every other trait. The remaining traits do not compete with each other and there is no migration between them. Therefore, the size of the subpopulation bearing a certain trait is just the equilibrium size of the trait if it is the only trait available, i.e. the one  determined by \eqref{eq:log}.

 \begin{figure}[hbtp!]
 \centering
\includegraphics[width=220pt]{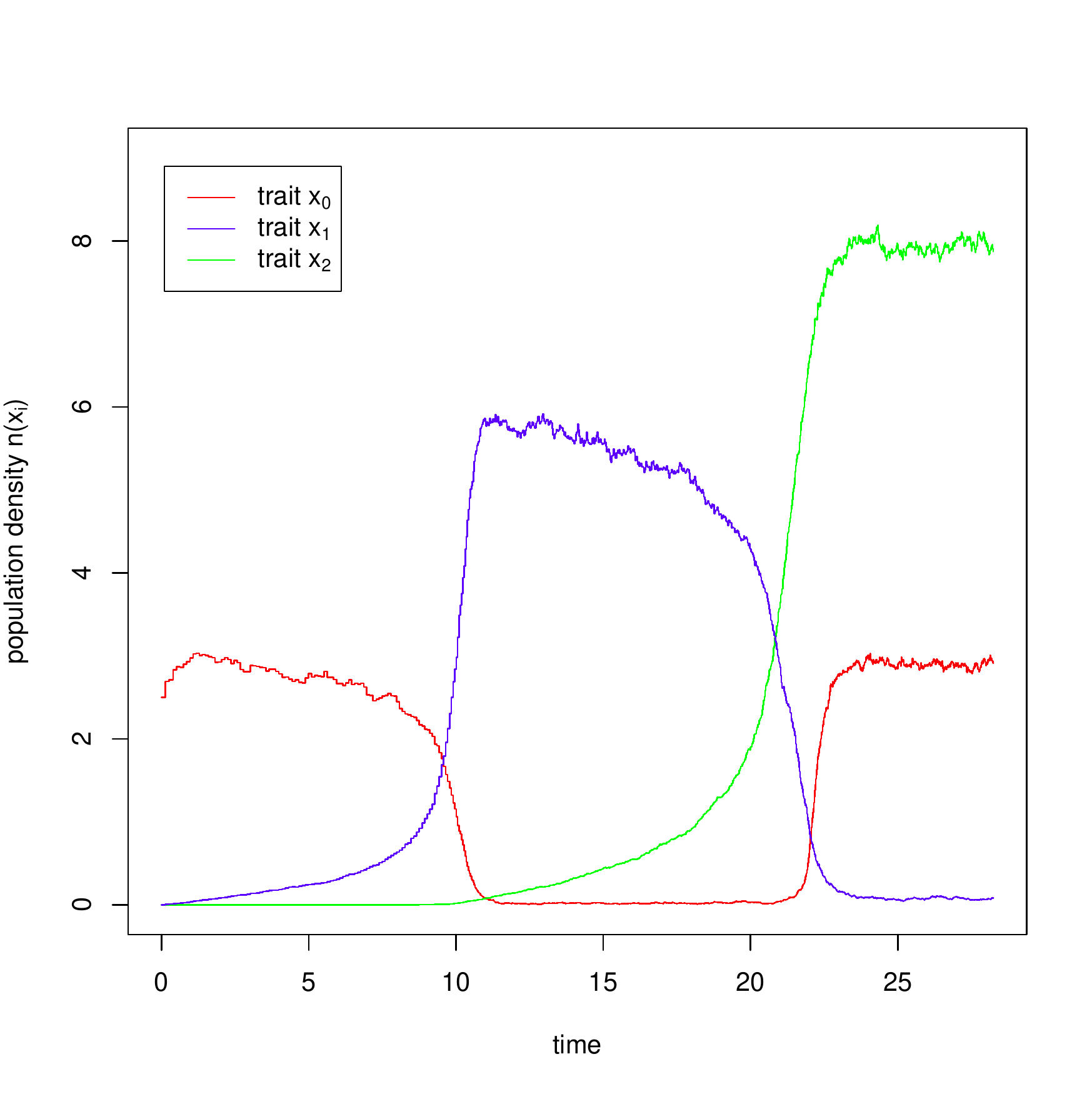}
\includegraphics[width=220pt]{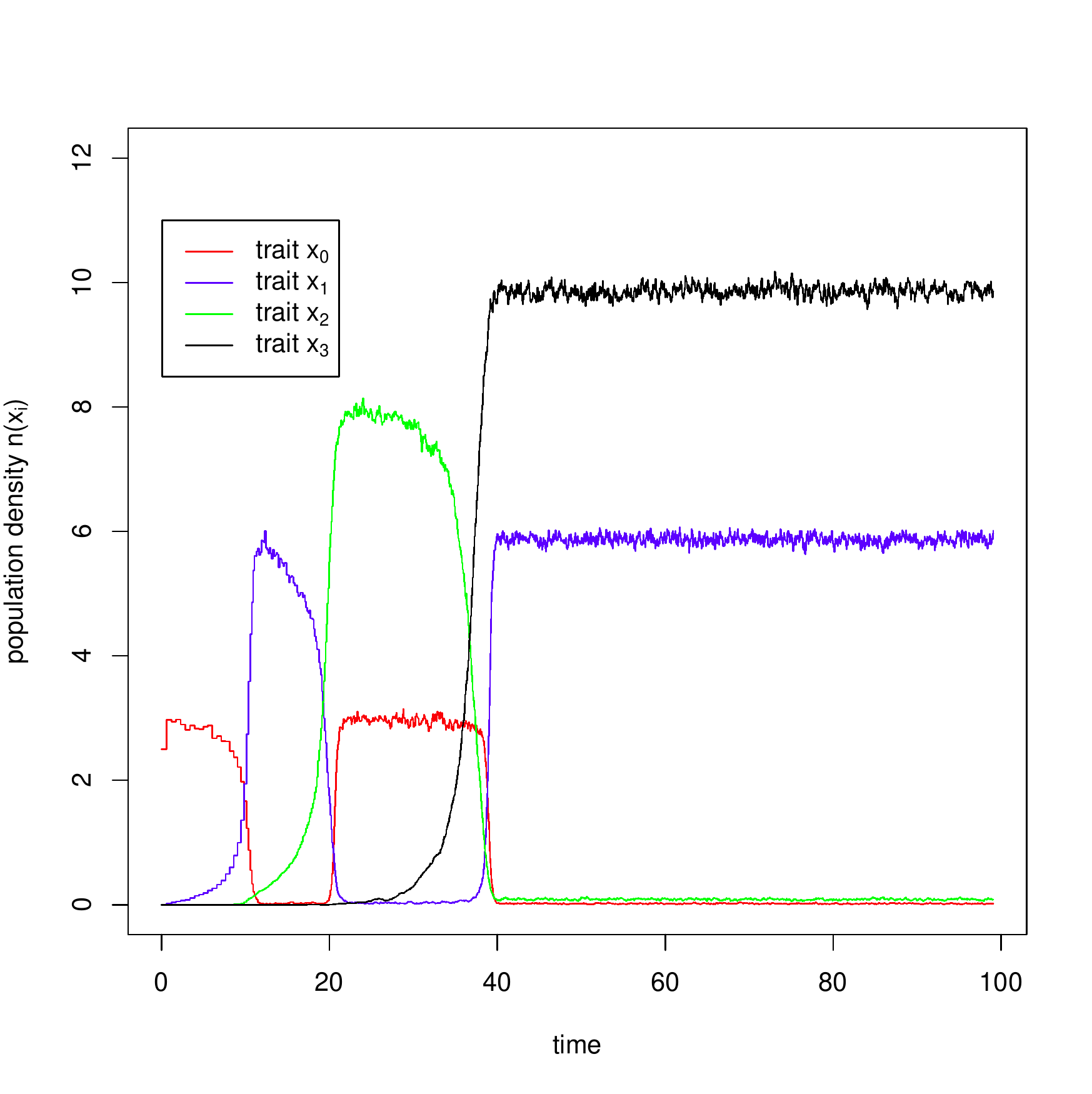}
\caption{\footnotesize Simulations of a development tree model arising in Proposition \ref{TST} on a three- and four-type trait space. }
\label{fig:TST}
 \end{figure}

Two numerical simulations are illustrated in Figure \ref{fig:TST}. The parameters are the same as for Figure \ref{fig:TSS} except that the initial population size is now $K=1000$. In addition,  $\epsilon=K^{-\frac{4}{5}}$ in the three-type case (LHS) and $\epsilon=K^{-\frac{3}{4}}$ in the four-type case (RHS). One can verify that condition \eqref{rare_migration} holds in both cases. According to Proposition \ref{TST}, the fixation timescale (i.e. the time after which the population can be approximated by $\Gamma^{(L)}$) is of order $\ln\frac{1}{\epsilon}$. The limit, stable configuration is $\Gamma^{(2)}=3\de_{x_0}+8\de_{x_2}$ for the three-type case and $\Gamma^{(3)}=6\de_{x_1}+10\de_{x_3}$ for the four-type case.

\subsection{Long-term development tree model with mutation} \label{subsec:migration2}
In this subsection, we incorporate mutation into the model and, in particular, focus on the joint effect of migration and mutation on different timescales. It turns out that a trait that almost dies out on the migration timescale has a chance to recover itself and further to be stabilized on the mutation timescale. This phenomenon results from a change in the fitness landscape, which is due to the arrival of a mutant.

As in Section \ref{subsec:TSS_mutation}, mutation takes place according to the mutation kernel $p(x,dh)$. Due to this law, a previously non-presenting trait appears from time to time, and enlarges the space of currently presenting traits. On the contrary, the migration kernel $m(x, dy)$ only acts on the space of currently presenting traits. The asymptotic behavior of the birth-death process with migration and mutation depends on the timescales of both events. Let us rescale the mutation law by $\sigma$ and the migration law by $\epsilon$.

First we need some assumptions and the definition of the trait substitution tree model, following \cite{BovierWang2013}.

\begin{assumption}\label{assp5}
1.  For any given distinct traits $\{x_0, x_1,\cdots, x_n\}\subset\X, n\in\N$, there exists a unique total order permutation
\begin{equation*}
x_{n_0}\prec x_{n_1}\prec\cdots\prec x_{n_{n-1}}\prec x_{n_n}.
\end{equation*}.

2. When there are $n+1$ distinct traits currently presenting, they are labeled by $x_0^{(n)}, x_1^{(n)} , \ldots, x_n^{(n)}$ such that $x_0^{(n)}\prec x_1^{(n)}\prec\cdots\prec x_n^{(n)}$. When a new trait $x$ appears, the $n+2$ presenting traits are relabeled by $x_0^{(n+1)}, x_1^{(n+1)} , \ldots, x_{n+1}^{(n+1)}$ such that $x^{(n+1)}_0\prec x^{(n+1)}_1\prec\cdots\prec x^{(n+1)}_n\prec x^{(n+1)}_{n+1}$.

3. Competition and migration only occurs between nearest neighbor and we force the interacting kernels to vanish between non-nearest neighbor sites, i.e. for all $n$, $m(x_i^{(n)},x_j^{(n)})=\al(x_i^{(n)},x_j^{(n)})\equiv 0$ for $\mid i-j\mid>1$. Hence, we retain the order $\prec$ as defined in Assumption \ref{assp3}.
\end{assumption}

At the end of the previous subsection, we point out that, on the migration timescale, there is a variety of paths to approach the equilibrium configuration. However, the equilibrium configuration of any $L$-trait system without mutation is always the same, i.e. $\Gamma^{(N)}$ defined by \eqref{eq:gamma}, and the timescale for convergence is always of order $O(\ln\frac{1}{\epsilon})$ as shown in Proposition \ref{TST}.
\begin{definition}\label{TST_Definition_infinite tree_micro}
A Markov jump process $\{\Gamma_t: t\geq 0\}$ is a \emph{trait substitution tree} (or TST for short) with ancestor $\Gamma_0=\bar n({x_0})\delta_{x_0}$ if the following hold:
\begin{description}
\item[(i)] For all $l\in \mathbb{N}$, all $k\in \{0, \ldots, l\}$, and all $h\in \X-x^{(2l)}_{2k}$, it jumps
    from configuration $\Gamma^{(2l)}=\sum_{i=0}^l\bar n(x^{(2l)}_{2i})\delta_{x^{(2l)}_{2i}}$, with transition rate $\bar n(x^{(2l)}_{2k})\mu(x^{(2l)}_{2k})p(x^{(2l)}_{2k},dh)$, to a new configuration $\Gamma^{(2l+1)}$ determined as follows:
    \begin{align*}
& \Gamma^{(2l+1)} =\\
& \begin{dcases}
 \sum_{i=1}^j \! \bar{n}(\!x^{(2l)}_{2i\!-\!1}\!) \delta_{x^{(2l)}_{2i\!-\!1}} \!+\! \bar{n}(\!x^{(2l)}_{2k}\!+\!h\!) \delta_{x^{(2l)}_{2k}\!+\!h} \!+\! \sum_{i=j+1}^l \! \bar{n}(\!x^{(2l)}_{2i}\!) \delta_{x^{(2l)}_{2i}}, \text{ if } \exists 0 \!\leq\! j \!\leq\! l \text{ s.t. } x^{(2l)}_{2j} \!\prec\! x^{(2l)}_{2k}\!+\!h \!\prec\! x^{(2l)}_{2j\!+\!1},\\
\sum_{i=1}^j \!\bar{n}(\!x^{(2l)}_{2i\!-\!1}\!) \delta_{x^{(2l)}_{2i\!-\!1}} \!+\! \sum_{i=j}^l \!\bar{n}(\!x^{(2l)}_{2i}\!) \delta_{x^{(2l)}_{2i}},  \hspace{3.5cm} \text{ if } \exists 0 \!\leq\! j \!\leq\! l \text{ s.t. } x^{(2l)}_{2j\!-\!1} \!\prec\! x^{(2l)}_{2k}\!+\!h \!\prec\! x^{(2l)}_{2j}.
\end{dcases}
 \end{align*}
  The traits presenting in $\Gamma^{(2l+1)}$ are relabeled according to the total order relation: $x_0^{(2l+1)}\prec x_1^{(2l+1)}\prec\cdots\prec  x_{2l}^{(2l+1)}\prec x_{2l+1}^{(2l+1)}$.

\item[(ii)] For all $l\in \mathbb{N}$, all $k\in \{0, \ldots, l+1\}$, and all $h\in \X-x^{(2l+1)}_{2k-1}$, it jumps from configuration $\Gamma^{(2l+1)}=\sum_{i=1}^{l+1}\bar n(x^{(2l+1)}_{2i-1})\delta_{x^{(2l+1)}_{2i-1}}$, with transition rate $\bar n(x^{(2l+1)}_{2k-1})\mu(x^{(2l+1)}_{2k-1})p(x^{(2l+1)}_{2k-1},dh)$, to a new configuration $\Gamma^{(2l+2)}$ determined as follows:
{\scriptsize
\begin{align*}
& \Gamma^{(2l+2)}=\\
& \begin{dcases}
\sum_{i=1}^j \!\bar{n}(x^{(\!2l\!+\!1\!)}_{2i\!-\!2\!}) \delta_{x^{(\!2l\!+\!1\!)}_{2i\!-\!2}} \!+\! \bar{n}(x^{(\!2l\!+\!1\!)}_{2k\!-\!1}\!+\!h) \delta_{x^{(\!2l\!+\!1\!)}_{2k\!-\!1}\!+\!h} + \sum_{i=j+1}^{l+1} \bar{n}(x^{(\!2l\!+\!1\!)}_{2i\!-\!1}) \delta_{x^{(\!2l\!+\!1\!)}_{2i\!-\!1}}, \text{ if } \exists 1 \!\leq\! j \!\leq\! l \!+\! 1 \text{ s.t. } x^{(\!2l\!+\!1\!)}_{2j\!-\!1} \!\prec\! x^{(\!2l\!+\!1\!)}_{2k\!-\!1} \!+\! h \prec x^{(\!2l\!+\!1\!)}_{2j},\\
\sum_{i=1}^j \bar{n}(x^{(\!2l\!+\!1\!)}_{2i\!-\!2}) \delta_{x^{(\!2l\!+\!1\!)}_{2i\!-\!2}} + \sum_{i=j}^{l+1} \bar{n}(x^{(\!2l\!+\!1\!)}_{2i\!-\!1}) \delta_{x^{(\!2l\!+\!1\!)}_{2i\!-\!1}}, \hspace{3.1cm}\text{ if } \exists 1 \!\leq\! j \!\leq\! l\!+\!1 \text{ s.t. } x^{(\!2l\!+\!1\!)}_{2j\!-\!2} \prec x^{(\!2l\!+\!1\!)}_{2k\!-\!1}\!+\!h \!\prec\! x^{(\!2l\!+\!1\!)}_{2j\!-\!1}.
\end{dcases}
\end{align*}}
  The traits presenting in $\Gamma^{(2l+2)}$ are relabeled according to the total order relation: $x_0^{(2l+2)}\prec x_1^{(2l+2)}\prec\cdots\prec  x_{2l+1}^{(2l+2)}\prec x_{2l+2}^{(2l+2)}$.
\end{description}
\end{definition}

Roughly speaking, when there are $n+1$ traits in the current trait space and ordered on an increasing fitness landscape, with the nearest-neighbor competition and migration rule, every other trait is absent from the current configuration. When a new trait $x$ appears by mutation, put it into the queue of the $n+1$ old traits according to its fitness. If the nearest trait with higher fitness than $x$ is absent (resp. present) in the current configuration, then trait $x$ survives (resp. vanishes) in the new configuration. In both cases, the old traits with higher fitness than $x$ remain present or absent in the new configuration as before, while the old traits with lower fitness than $x$ switch their presence to absence or vice verse in the new configuration.

\begin{proposition}\label{TST_infinite_trait}
Consider the process $\{X^{K,\epsilon,\sigma}_t\}_{t\geq 0}$ with rescaled migration law $\epsilon m(x,y)$ and rescaled mutation law $\sigma \mu(x)$. Suppose that for all $K$, $X_0^{K,\epsilon,\sigma}= \frac{N^K_0}{K}\delta_{x_0}$, and $\frac{N^K_0}{K}\to \bar n(x_0)$ in law as $K\to\infty$. If, in addition to condition \eqref{rare_migration} that $1\ll K\epsilon\ll K$, one has
\begin{equation}\label{rare_mutation_condition}
\ln\frac{1}{\epsilon}\ll\frac{1}{K\sigma}\ll e^{KC}, \quad \forall  C>0,
\end{equation}
then $(X^{K,\epsilon,\sigma}_{t/K\sigma})_{t\geq 0}$ converges, as $K\to \infty$, to the trait substitution tree $(\Gamma_{t})_{t\geq 0}$ defined in Definition \ref{TST_Definition_infinite tree_micro} in the sense of f.d.d. on $\mathcal{M}_F(\mathcal{X})$ equipped with the topology induced by mappings $\nu\mapsto\langle\nu,f\rangle$ with $f$ a bounded measurable function on $\X$..
\end{proposition}

The proof is in the Appendix. According to \eqref{rare_mutation_condition}, the fixation timescale of order $\ln\frac{1}{\epsilon}$ is much shorter than the  mutation timescale of order $\frac{1}{K\sigma}$. At the same time, $\frac{1}{K\sigma} \ll e^{KC}$ prevents the system from drifting away from the TST equilibrium configuration on the mutation timescale \citep{FW84}.

 \begin{figure}[hbtp!]
 \centering
\includegraphics[width=220pt]{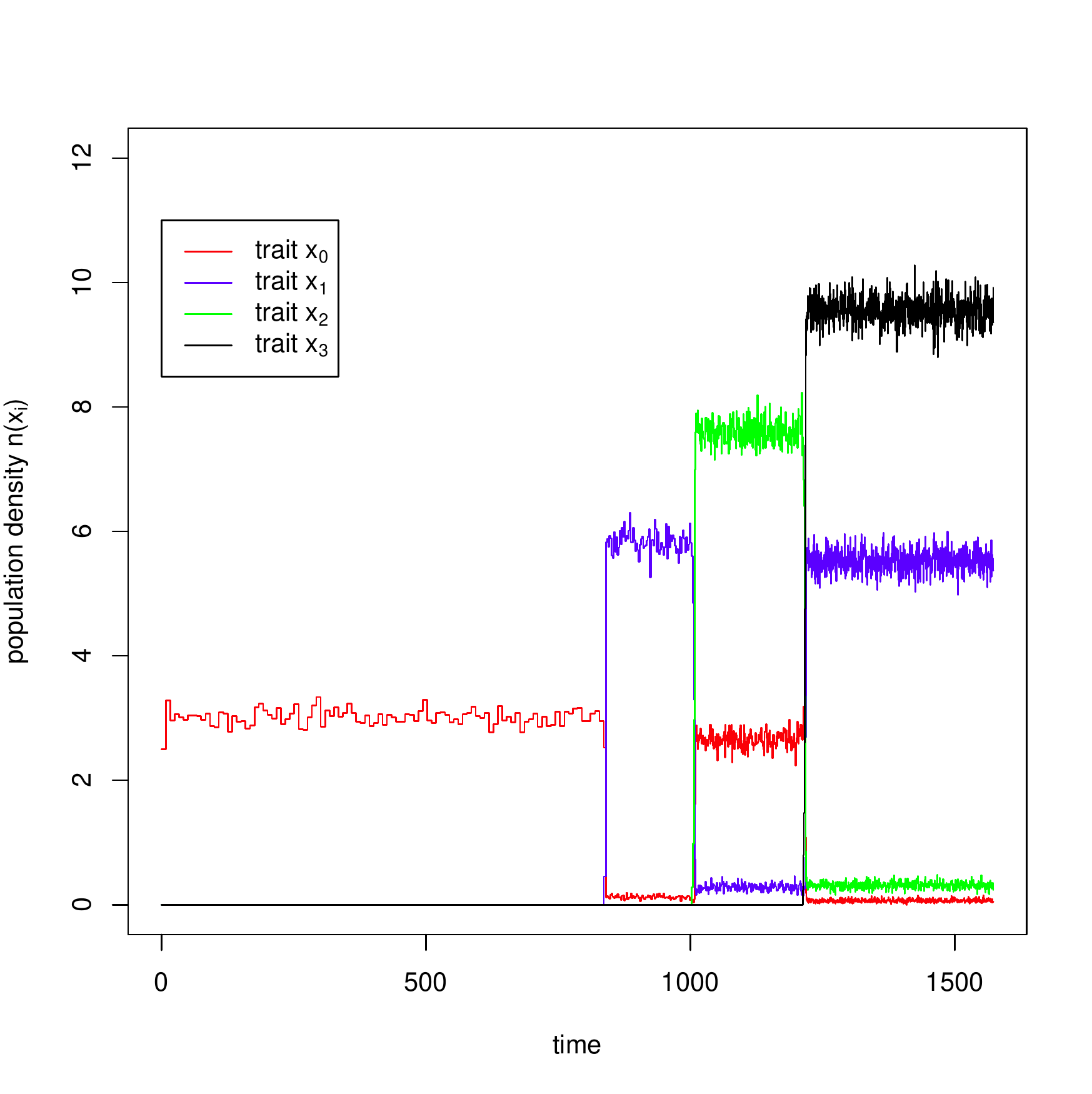}
\includegraphics[width=220pt]{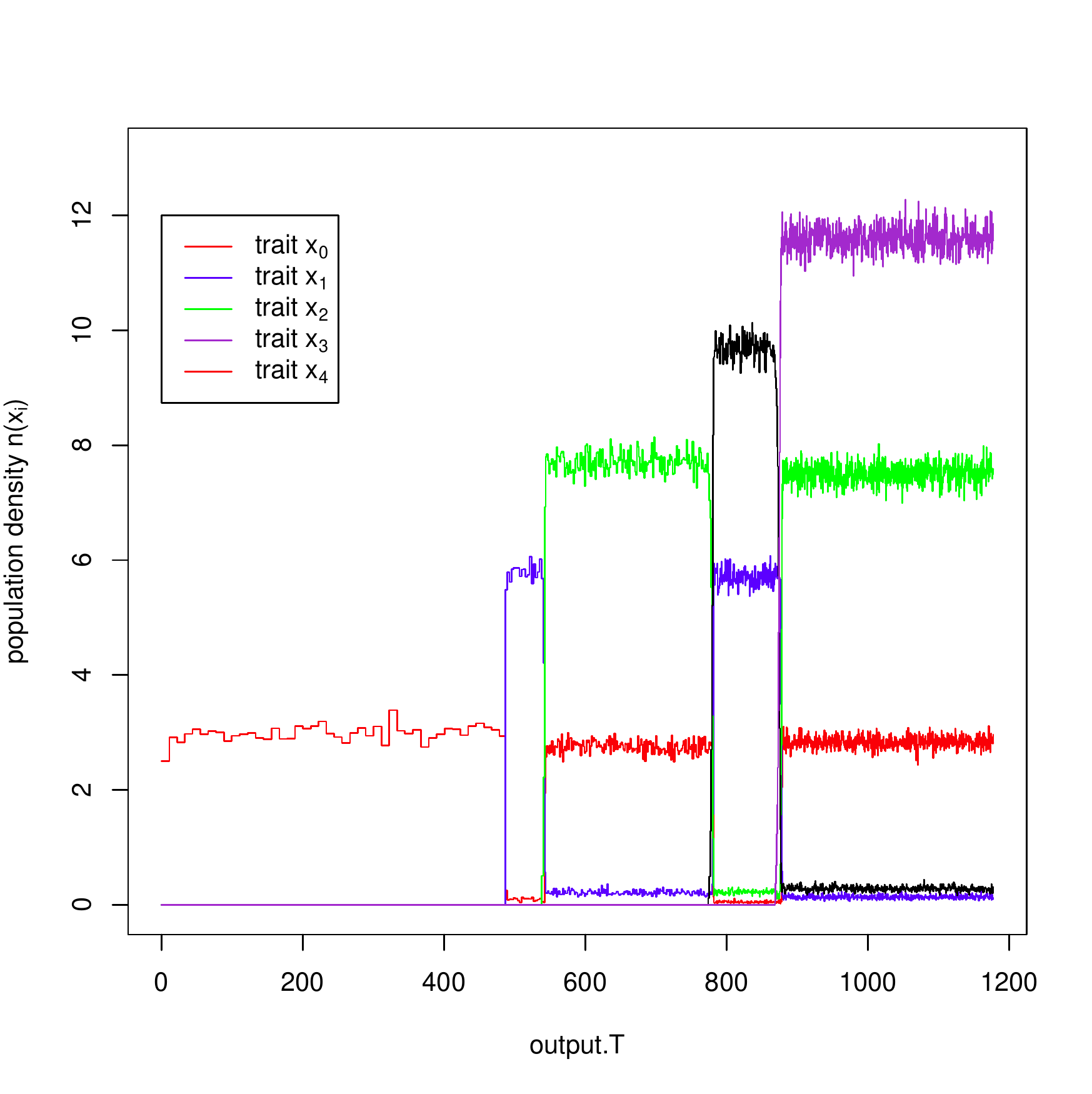}
\caption{\footnotesize Simulations of a trait substitution tree on the mutation timescale arising in Proposition \ref{TST_infinite_trait} on four- and five-type trait space.}
\label{TST_mutation}
 \end{figure}

Return to the examples shown in Figure \ref{fig:TST} and add mutations to them. Let us consider the particular case where a mutant is always fitter than all the existing traits. In both examples, the birth rates of the red, the blue, the green, the black and the purple-colored populations are respectively 3, 6, 8, 10, and 12, while their death rates are all constantly $0$. Take $\epsilon=K^{-0.8}$ and $\sigma=K^{-1.5}$, and the initial scaling parameter $K=400$. The simulation results are shown in Figure \ref{TST_mutation}. Remark how the arrival of each mutant ``reshuffles the cards'' by initiating a new round of competition which leads to a new equilibrium. Let us compare these examples with the ones shown in Figure \ref{fig:TSS}, where no migration is concerned. One sees that, with migration, a once-eliminated trait $x$ now has a chance to be revived, because its competitor $y$ which eliminated $x$ in the previous round may itself now be eliminated by some other trait. Observe that the TST process jumps from $\Gamma^{(n)}$ to $\Gamma^{(n+1)}$ on a mutation timescale of order $\frac{1}{K\sigma}$. Besides, on this timescale, the fixation process (describing short-term pair-wise competition) is not visible anymore. However, by zooming into the infinitesimal fixation period, $S$-type curves as in Figure \ref{fig:TST} will still emerge.


The specific rule of nearest-neighbor migration and competition is adopted in this section as an example to illustrate the development tree model. When the model is applied to a particular problem, the local migration and competition rules should be determined according to the context. Although the specific form of the process limits in Propositions \ref{TST} and \ref{TST_infinite_trait} (every other trait survives on the increasing fitness landscape) may no longer hold, it gives us some flavor of the tree form depicting a long-term dynamic, where short-term competitions and rare mutations lead to diversification and development.


\section{Discussion}
A notable feature of our tree model is the three distinct timescales related with different events. The shortest one is the lifecycle time of a single generation, which is of order $1$. The intermediate timescale is identified as the migration time of order $\ln \frac{1}{\epsilon}$. The relation between these two timescale (cf. equation \eqref{rare_migration}) justifies the macroscopic approximation of a population. Models limited to these two timescales have been applied in the framework of short-term competition/substitution, which leads to a temporary equilibrium. The novelty of our model is to add a third, the longest timescale of order $\frac{1}{K\sigma}$, which corresponds to the interval between two mutations. Observing from this timescale, the process is composed of a series of rapid transitions from one temporary equilibrium to another. Indeed, every temporary equilibrium, though lasting a long time, is eventually broken by the arrival of a mutant or an innovation. These three timescales of distinct orders allow us to draw a tree form graph depicting the long-term development of the alternatives in question.

%

Besides, though our model aims to explain the long-term development and diversification of technologies, ideas, custom etc. in various contexts, it remains an illustrative tool on the theoretical level. For example, the events such as birth, death and migration have different signification in different context. We have adopted an evolutionary setting in this paper so that the parameters such as birth-death rates or migration laws are exogenously given. There is however a large body of literature that focuses on the strategic side of individual decision making, and builds models with different mechanisms generating rational choices or imitative behavior. The usual approaches consist of non cooperative games (with decision-making on an individual level) and social learning (at the population level). The former often takes the uncertainty and heterogeneity on the individual level into account.\footnote{An incomplete list of references in economics includes \cite{Hiebert1974}, \cite{Stone1981}, \cite{Jensen1982}, \cite{ChaElia1990}, \cite{SinChan1992}, \cite{Meade1989}, \cite{BEMS1996}, \cite{CER2000}, \cite{Young2009}, \cite{Young2011} and \cite{KY2014}.} Integrating these approaches in our specification constitutes an obvious avenue for future research.



\section*{Appendix}

\begin{proof}[Proof of Proposition \ref{prop:stable analysis of LV}]
In fact, there are four fixed points of the two-dimensional L-V system \eqref{eq:lv}, namely, $(0,0)$, \,$(\bar n(x),0)$, \,$(0, \bar n(y))$, and $(n^*(x), n^*(y))$, where $(n^*(x), n^*(y))$ is such that

\begin{equation}\nonumber
\left\{
  \begin{array}{ll}
    b(x)-d(x)-\alpha(x,x)n_t(x)-\alpha(x,y)n_t(y)=0 \\
    b(y)-d(y)-\alpha(y,x)n_t(x)-\alpha(y,y)n_t(y)=0.
  \end{array}
\right.
\end{equation}

By simple calculation, we obtain that
\begin{equation}\nonumber
\left\{
  \begin{array}{ll}
    n^*(x)=\frac{\al(y,y)f(x,y)}{\al(x,x)\al(y,y)-\al(x,y)\al(y,x)} \\
    n^*(x)=\frac{\al(x,x)f(y,x)}{\al(x,x)\al(y,y)-\al(x,y)\al(y,x)}.
  \end{array}
\right.
\end{equation}
To make sense of the solution as a population density (which must be non-negative), one needs $f(x,y)\cdot f(y,x)>0$. It contradicts the assumption $f(x,y)<0, f(y,x)>0$. We thus exclude the solution $(n^*(x), n^*(y))$.

The Jacobian matrix for the system \eqref{eq:lv} at point $(0,0)$ is
\begin{equation}\nonumber
\left(
  \begin{array}{cc}
    b(x)-d(x) & 0 \\
    0 & b(y)-d(y) \\
  \end{array}
\right).
\end{equation}
Obviously its eigenvalues are both positive. Thus $(0,0)$ is unstable.

The Jacobian matrix at point $(\bar n(x), 0)$ is
\begin{equation}\nonumber
\begin{aligned}
\left(
  \begin{array}{cc}
    -\left(b(x)-d(x)\right) & -\al(x,y)\bar n(x) \\
    0 & b(y)-d(y)-\al(y,x)\bar n(x) \\
  \end{array}
\right)\\
=\left(
   \begin{array}{cc}
      -\left(b(x)-d(x)\right) & -\al(x,y)\bar n(x) \\
      0 & f(y,x) \\
   \end{array}
 \right).
\end{aligned}
\end{equation}
Since one of its eigenvalue $-\left(b(x)-d(x)\right)$ is negative whereas the other one is $f(y,x)>0$, the equilibrium $(\bar n(x), 0)$ is unstable.

The Jacobian matrix of system \eqref{eq:lv} at point $(0, \bar n(y))$ is
\begin{equation}\nonumber
\left(
  \begin{array}{cc}
    b(x)-d(x)-\al(x,y)\bar n(y) & 0 \\
    -\al(y,x)\bar n(y) & -\left(b(y)-d(y)\right) \\
  \end{array}
\right)\\
=\left(
   \begin{array}{cc}
     f(x,y) & 0 \\
     -\al(y,x)\bar n(y) & -\left(b(y)-d(y)\right) \\
   \end{array}
 \right),
\end{equation}
whose eigenvalues are both negative because of the condition $f(x,y)<0$.
Thus $(0, \bar n(y))$ is the only stable equilibrium of the system \eqref{eq:lv}.
\end{proof}

\begin{proof}[Proof of Proposition \ref{TST}]
Let us prove the result for the three-trait toy model (see Figure \ref{Figure_Phase3type_Micro}) to illustrate the basic idea of proof, though our analysis is not limited to the three-trait case only. Indeed, the whole machinery is available for any finite-trait space.

Assume  $\X=\{x_0, x_1, x_2\}$.
Let $\xi_t^K(x_0):=\langle X_t^K, 1_{\{x_0\}}\rangle$ and $\xi_t^K(x_i):=\langle X_t^K, 1_{\{x_i\}}\rangle$ for $i=1,2$.

 \begin{figure}[hbtp]
 \centering
 \def\svgwidth{400pt}
 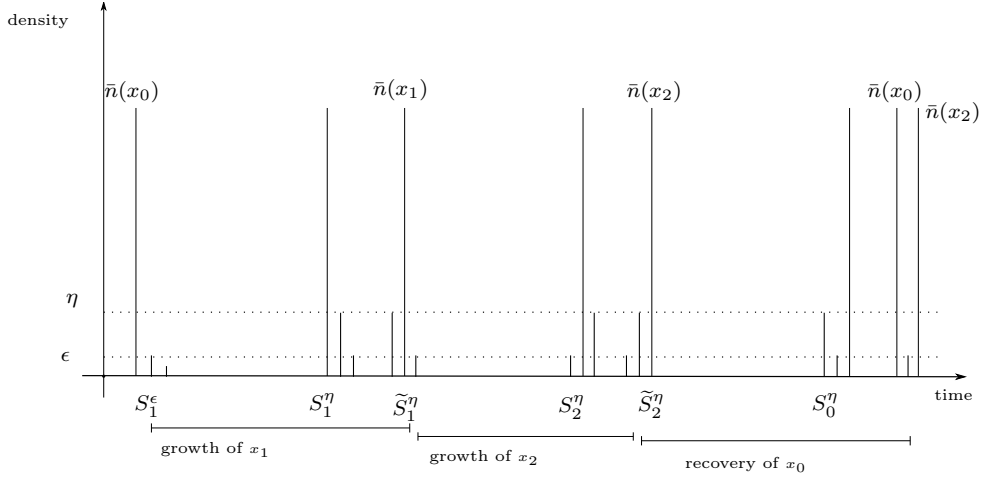
 \caption{{\footnotesize Phase evolution of mass bars in early time window on the three-trait site space}}.
 \label{Figure_Phase3type_Micro}
 \end{figure}

\textbf{Step 1}. Firstly, consider the emergence and growth of population at trait site $x_1$. Set $S_1^{\epsilon}=\inf\{t>0: \xi_t^K(x_1)\geq \epsilon\}$. Thanks to $\frac{N_0^K}{K}\to n_0>0$ in law as $K\to\infty$ and by applying the Law of Large Numbers of random processes (see Chap.11, Ethier and Kurtz 1986), one obtains that, for any $\de>0,\,T>0$,
\[
\lim\limits_{K\to\infty}\PP\left(\sup\limits_{0\leq t\leq T}\left|\frac{\xi_t^K(x_1)}{\epsilon}-n_t(x_1)\right|<\de\right)=1
\]
where $n_t(x_1)$ is governed by equation $\dot{n}(x_1)= m(x_0,x_1)n_0$ with initial $n_0(x_1)=0$. Therefore,
\begin{equation}\label{linear growth before epsilon threshold}
\lim\limits_{K\to\infty}\PP\left(\frac{1}{m(x_0,x_1)n_0}-\de<S_1^{\epsilon}<\frac{1}{m(x_0,x_1)n_0}+\de\right)=1,
\end{equation}
that is, $S_1^{\epsilon}$ is of order 1.

For any $\eta>0$, set $S_1^{\eta}=\inf\{t: t>S_1^{\epsilon},\,\xi_t^K(x_1)\geq \eta\}$. Consider a sequence of rescaled processes $\left(\frac{N_t^{K,1}}{K\epsilon}\right)_{t\geq S_1^{\epsilon}}$ with $\frac{N_{S_1^{\epsilon}}^{K,1}}{K\epsilon}=\frac{\xi_{S_1^{\epsilon}}^K(x_1)}{\epsilon}\to 1$ as $K\to\infty$. As before, by law of large numbers of random processes (see Chap.11 Ethier and Kurtz 1986), one obtains, for any $\de>0, \,T>0$,
\begin{equation}\label{lln before eta threshold}
\lim\limits_{K\to\infty}\PP\left(\sup\limits_{0\leq t\leq T}\left|\frac{N_t^{K,1}}{K\epsilon}-m_t\right|<\de\right)=1,
\end{equation}
where $m_t$ is governed by equation $\dot{m}=\bar{f}(x_1,x_0)m=\left(b(x_1)-d(x_1)-\al(x_1,x_0)\bar{n}(x_0)\right)m$ with $m_0=1$.

Set $T_1^{\eta/\epsilon}=\inf\{t-S_1^{\epsilon}: t>S_1^{\epsilon},\,\frac{N_t^{K,1}}{K\epsilon}\geq \eta/\epsilon\}$, and $t_1^{\eta/\epsilon}=\inf\{t>0: m_t\geq \eta/\epsilon\}$.
Then, for any $\de>0$, there exists $\de^{'}>0$ such that
\begin{equation}\label{time estimate before eta threshold}
\begin{aligned}
\lim\limits_{K\to\infty}&\PP\left(\left(\frac{1}{\bar{f}(x_1,x_0)}-\de\right)\ln\frac{1}{\epsilon}< S_1^\eta-S_1^{\epsilon}<\left(\frac{1}{\bar{f}(x_1,x_0)}+\de\right)\ln\frac{1}{\epsilon}\right)\\
=\lim\limits_{K\to\infty}&\PP\left(\left(\frac{1}{\bar{f}(x_1,x_0)}-\de\right)\ln\frac{1}{\epsilon}< T_1^{\eta/\epsilon}<\left(\frac{1}{\bar{f}(x_1,x_0)}+\de\right)\ln\frac{1}{\epsilon}\right)\\
=\lim\limits_{K\to\infty}&\PP\Bigg(\left(\frac{1}{\bar{f}(x_1,x_0)}-\frac{\de}{2}\right)\ln\frac{1}{\epsilon}< t_1^{\eta/\epsilon}<\left(\frac{1}{\bar{f}(x_1,x_0)}+\frac{\de}{2}\right)\ln\frac{1}{\epsilon},\\&\sup\limits_{0\leq t\leq t_1^{\eta/\epsilon}}\mid \frac{N_t^{K,1}}{K\epsilon}-m_t\mid<\de^{'}\Bigg)\\
=&1
\end{aligned}
\end{equation}
where the last equal sign is due to \eqref{lln before eta threshold}.

After population of trait $x_1$ reaches some $\eta$ threshold, the dynamics $\left(\xi_t^K(x_0),\xi_t^K(x_1)\right)$ can be approximated by the solution of a two-dimensional Lotka-Volterra equations. By Proposition \ref{prop:stable analysis of LV}, it takes time of order 1 (mark this time coordinator by $\widetilde S_1^{\eta}$) for the two subpopulations switching their mass distribution and gets attracted into $\eta-$neighborhood of the stable equilibrium $\left(0,\bar{n}(x_1)\right)$.

\textbf{Step 2.} Now consider the emerging and growth of population $\xi_t^K(x_2):=\langle X_t^K, 1_{\{x_2\}}\rangle$ at trait site $x_2$. Set $S_2^{\epsilon}=\inf\{t: t>\widetilde{S}_1^{\eta}, \,\xi_t^K(x_2)\geq \epsilon\}$.
Similarly as is done for $S_1^{\epsilon}$ in \eqref{linear growth before epsilon threshold}, one can get that $\lim\limits_{K\to\infty}\PP(S_2^{\epsilon}-\widetilde{S}_1^{\eta}=O(1))=1$. On a longer time scale, we will not distinguish $S_2^{\epsilon}$ from $\widetilde{S}_1^{\eta}$.

Set $S_2^{\eta}=\inf\{t: t>S_2^{\epsilon},\, \xi_t^K(x_2)\geq \eta\}$. One follows the same procedure to derive \eqref{time estimate before eta threshold} and asserts that for any $\de>0$,
\begin{equation}\label{time estimate before second eta threshold}
\lim\limits_{K\to\infty}\PP\left(\left(\frac{1}{\bar{f}(x_2,x_1)}-\de\right)\ln\frac{1}{\epsilon}< S_2^\eta-\widetilde S_1^{\eta}<(\frac{1}{\bar{f}(x_2,x_1)}+\de)\ln\frac{1}{\epsilon}\right)
=1.
\end{equation}
Note that assumption (\textbf{B}3) $\frac{2}{b(x_2)-d(x_2)}\geq \frac{1}{\bar f(x_1,x_0)}+\frac{1}{\bar f(x_2,x_1)}$ guarantees that $\xi^K_t(x_2)$ can not grow so fast in exponential rate $b(x_2)-d(x_2)$ such that it reaches some $\eta$-level before $S_2^{\eta}$.

During time period $(\widetilde S_1^{\eta}, S_2^{\eta})$, population at site $x_0$, on one hand, decreases due to the competition from more fitter trait $x_1$. On the other hand, it can not go below $\epsilon$ level due to the successive migration in a portion of $\epsilon$ from site $x_1$. More precisely, by neglecting migrant contribution,
$\xi_t^K(x_0)$ converges $n_t(x_0)$ in probability as $K$ tends to $\infty$, where
\begin{equation}
\dot{n}_t(x_0)=\left(b(x_0)-d(x_0)-\al(x_0,x_1)\bar{n}(x_1)\right) n_t(x_0)=\bar{f}(x_0,x_1)n_t(x_0)
\end{equation}
with $n_0(x_0)=\eta$. Let $\Delta S_2^\eta=S_2^\eta-\widetilde S_1^{\eta}$. Then, for any $\de>0$,
\begin{equation}
\begin{aligned}
&\lim\limits_{K\to\infty}\PP\left(\xi_{S_2^{\eta}}^K(x_0)\in(n_{\Delta
  S_2^{\eta}}(x_0)-\de,n_{\Delta S_2^{\eta}}(x_0)+\de)\right)\\
&=\lim\limits_{K\to\infty}\PP\left(\eta
  e^{\bar{f}(x_0,x_1)\Delta S_2^{\eta}}-\de<\xi_{S_2^{\eta}}^K(x_0)<\eta e^{\bar{f}(x_0,x_1)\Delta S_2^{\eta}}+\de\right)\\
&=\lim\limits_{K\to\infty}\PP\left(\eta
  \epsilon^{\mid\bar{f}(x_0,x_1)\mid/\bar{f}(x_2,x_1)}-\de<\xi_{S_2^{\eta}}^K(x_0)<\eta \epsilon^{\mid\bar{f}(x_0,x_1)\mid/\bar{f}(x_2,x_1)}+\de\right)\\
&=1
\end{aligned}
\end{equation}
where the second equality is due to \eqref{time estimate before second eta threshold}.
Taking the migration from site $x_1$ into account, we thus have
\begin{equation}\label{density estimate original type after second eta threshold}
\lim\limits_{K\to\infty}\PP\left(\xi_{S_2^{\eta}}^K(x_0)=O(\epsilon^{\mid\bar{f}(x_0,x_1)\mid/\bar{f}(x_2,x_1)}
\vee \epsilon)\right)=1.
\end{equation}
We proceed as before for $\widetilde S_1^{\eta}$ in step 1. After time $S_2^{\eta}$, the mass bars on dimorphic system $(\xi_t^K(x_1),\xi_t^K(x_2))$ can be approximated by ODEs and will be switched again in time of order 1 (marked by $\widetilde S_2^{\eta}$ as in Figure \ref{Figure_Phase3type_Micro}), and they are attracted into $\eta-$ neighborhood of $(0,\bar{n}(x_2))$. As for the population density on site $x_0$, one obtains from \eqref{density estimate original type after second eta threshold}
\begin{equation}
\lim\limits_{K\to\infty}\PP\left(\xi^K_{\widetilde{S}^{\eta}_2}(x_0)=O(\epsilon^{c_1})\right)=1
\end{equation}
where $c_1=\frac{|\bar f(x_0,x_1)|}{\bar f(x_2,x_1)}\wedge 1\leq 1$.

\textbf{Step 3.} We now consider the recovery of subpopulation at trait site $x_0$. Recovery arises because of the lack of effective competitions from its neighbor site $x_1$, or under negligible competitions since the local population density on $x_1$ is very low under the control of its fitter neighbor $x_2$.
Without loss of generality, we suppose $c_1:=\frac{\mid\bar{f}(x_0,x_1)\mid}{\bar f(x_2,x_1)}<1$ in \eqref{density estimate original type after second eta threshold}.

Set $S_0^{\eta}=\inf\{t: t>\widetilde S_2^{\eta}, \,\xi_t^K(x_0)\geq \eta\}$. We proceed as before in step 1. From \eqref{density estimate original type after second eta threshold}, $\frac{\xi_{\widetilde S_2^{\eta}}^K(x_0)}{\epsilon^{c_1}}$ converges to some positive constant (say $m_0$) in probability as $K\to\infty$. Thus, by applying law of large numbers to the sequence of processes $\frac{N^K_t}{K\epsilon^{c_1}}$, for any $\de>0,\,T>0,$
\begin{equation}
\lim\limits_{K\to\infty}\PP\left(\sup\limits_{0\leq t\leq T}\left|\frac{\xi_t^K(x_0)}{\epsilon^{c_1}}-m_t\right|<\de\right)=1
\end{equation}
where $m_t$ is governed by logistic equation $\dot{m}=\left(b(x_0)-d(x_0)\right)m$ starting with a positive initial $m_0$.

Following the same way to obtain \eqref{time estimate before eta threshold}, time length $S_0^{\eta}-\widetilde S_2^{\eta}$ can be approximated by time needed for dynamics $m$ to approach $\eta/\epsilon^{c_1}$ level, which is of order $\frac{c_1}{(b(x_0)-d(x_0))}\ln\frac{1}{\epsilon}$, i.e. for any $\de>0$,
\begin{equation}\label{time estimate for original recovery}
\lim\limits_{K\to\infty}\PP\left(\left(\frac{c_1}{b(x_0)-d(x_0)}-\de\right)\ln\frac{1}{\epsilon}< S_0^\eta-\widetilde S_2^{\eta}<\left(\frac{c_1}{b(x_0)-d(x_0)}+\de\right)\ln\frac{1}{\epsilon}\right)
=1.
\end{equation}
At the same time, $\xi^K_t(x_1)$ converges in probability to $\psi_t$ which satisfies equation $\dot{\psi}=\bar{f}(x_1,x_2)\psi$ with $\psi_{\widetilde{S}^{\eta}_2}=\eta$. Then, we can justify the following estimate for population density at site $x_1$,
\begin{equation}\label{density estimate first type after original recovery}
\lim\limits_{K\to\infty}\PP\left(\xi_{S_0^{\eta}}^K(x_1)=O(\epsilon^{c_2}
\vee \epsilon)\right)=1
\end{equation}
where $c_2=\frac{c_1\mid\bar{f}(x_1,x_2)\mid}{b(x_0)-d(x_0)}$.

We now combine all these estimates \eqref{time estimate before eta threshold}, \eqref{time estimate before second eta threshold}, \eqref{time estimate for original recovery} together, and conclude that
\begin{equation}
\lim\limits_{K\to\infty}\PP\left(\| X^K_{t\ln\frac{1}{\epsilon}}-\Gamma^{(2)}\|<\de\right)=1
\end{equation}
for $t>\bar{t}_2:=\frac{1}{\bar{f}(x_1,x_0)}+\frac{1}{\bar{f}(x_2,x_1)}+\frac{c_1}{b(x_0)-d(x_0)}$ under the total variation norm $\|\cdot\|$ on $\mathcal{M}_F(\X)$.
\end{proof}

\begin{proof}[Proof of Proposition \ref{TST_infinite_trait}]
First recall that all Markov jump processes can be specified by two features: the exponentially distributed waiting time, and the one-step transition rule (\cite{EK1986}). Our proof, which shows that the limiting process is a Markov jump process $\Gamma_t$, is hence made up of two parts.

The first part of the proof consists in the characterization of exponential waiting time of each mutation arrival, as can be seen from the construction of the process in Section \ref{section:preliminary}. Let $L=2l$, $X_0^{K,\epsilon,\sigma}=\Gamma^{(L)}$, and $\tau$ be the first mutation arrival time after time 0. Then, similar arguments used in \cite[Lemma 2 (c)]{Champnt2006} shows us that
\begin{lemma}
\begin{equation}\label{eq:33}
\lim\limits_{K\to\infty} \PP\left(\tau>\frac{t}{K\sigma}\right)=\exp\left(-t\sum_{i=0}^l\bar n(x_{2i}^{(2l)})\mu(x_{2i}^{(2l)})\right),
\end{equation}
and
\begin{equation}\label{eq:34}
 \lim\limits_{K\to\infty} \PP\left(\textrm{at time}\,\tau,\, \textrm{mutant comes from trait}~  x_{2k}^{(2l)}\right)=\frac{\bar n(x_{2k}^{(2l)})\mu(x_{2k}^{(2l)})}{\sum_{i=0}^l\bar n(x_{2i}^{(2l)})\mu(x_{2i}^{(2l)})}.
 \end{equation}
 \end{lemma}
 We will not repeat the details of the proof.

The second part of the proof can been seen as a corollary of Proposition \ref{TST}. It specifies the new equilibrium configuration and shows that fixation time of the new configuration is of order $\ln\frac{1}{\epsilon}$, which is invisible on the mutation timescale.

\begin{lemma}
Assume that $X_0^{K,\epsilon,\sigma}=\Gamma^{(2l)}+\frac{1}{K}\delta_{x_{2k}^{(2l)}+h}$ for some $0\leq k\leq l$. Then there exists a constant $C>0$, for any $\delta>0$,  such that
\begin{equation}
\lim\limits_{K\to\infty} \PP\Bigl(\tau>C\ln\frac{1}{\epsilon},\,\sup\limits_{t\in(C\ln\frac{1}{\epsilon}, \tau)}\|X_t^{K,\epsilon, \sigma}-\Gamma^{(2l+1)}\|<\delta\Bigr)=1
\end{equation}
where $\Gamma^{(2l+1)}$ associated to Definition \ref{TST_Definition_infinite tree_micro} (i) is defined as the following, in the first case
\begin{equation*}
x_i^{(2l+1)}=x_i^{(2l)}\; \forall \,0\leq i\leq 2j,\;\, x^{(2l+1)}_{2j+1}=x^{(2l)}_{2k}+h,\;\, x_i^{(2l+1)}= x_{i-1}^{(2l)}\; \forall \,2j+2\leq i\leq 2l+1;
 \end{equation*}
 in the second case
\begin{equation*}
x_i^{(2l+1)}=x_i^{(2l)}\; \forall \,0\leq i\leq 2j-1,\,\; x^{(2l+1)}_{2j}=x^{(2l)}_{2k}+h,\;\, x_i^{(2l+1)}= x_{i-1}^{(2l)} \; \forall \,2j+1\leq i\leq 2l+1.
\end{equation*}
\end{lemma}
\begin{proof}[Proof of the lemma]
Indeed, equation \eqref{eq:33} implies that for all $C>0$,
\begin{equation}\lim\limits_{\epsilon\to 0}\PP(\tau^{\epsilon}>C\ln\frac{1}{\epsilon})=1.
\end{equation}

According to the fitness landscape, there is one and only one ordered position for the new arising trait $x_{2k}^{(2l)}+h$ in $\Gamma^{(2l)}$. Suppose there exists $j\in \{0,1,\ldots, l\}$ such that $x_{2k}^{(2l)}+h$ fits between $x_{2j}^{(2l)}$ and $x_{2j+1}^{(2l)}$, i.e.
 \begin{equation}
 x_{2j-1}^{(2l)}\prec x_{2j}^{(2l)}\prec x_{2k}^{(2l)}+h\prec x_{2j+1}^{(2l)}.
\end{equation}

Since both traits $x_{2j-1}^{(2l)}$ and $x_{2j+1}^{(2l)}$ are absent in $\Gamma^{(2l)}$, the pair $(x_{2j}^{(2l)},\, x_{2k}^{(2l)}+h)$ is isolated  and hence without competition from others. By the same argument as in the proof of Proposition \ref{TST}, the two-type system converges to $(0,\,\bar n(x_{2k}^{(2l)}+h)\delta_{x_{2k}^{(2l)}+h})$ in time of order $O(\ln\frac{1}{\epsilon})$. For those traits of higher fitness than the isolated pair, nothing changes due to their isolation. Whereas for the traits of lower fitness than the pair, trait $x_{2j-1}^{(2l)}$ increases exponentially due to the decay of its fitter neighbor $x_{2j}^{(2l)}$. So on and so forth, the mass occupation flips on the left hand side of $x_{2j}^{(2l)}$. The entire rearrangement process takes time of order $O(\ln\frac{1}{\epsilon})$.

In a similar way, we can prove the case where there is $j\in \{0,1,\ldots, l\}$ such that $x_{2j-1}^{(2l)}\prec x_{2k}^{(2l)}+h\prec x_{2j}^{(2l)}\prec x_{2j+1}^{(2l)}$.
\end{proof}
The case where $L=2l+1$ can be proved similarly.
\end{proof}

\bibliographystyle{ecca}
\bibliography{tstbib}

\end{document}